\documentclass[11pt,a4paper]{article}

\usepackage[utf8]{inputenc}
\usepackage[T1]{fontenc}
\usepackage{amsmath,amssymb,amsfonts,amsthm}
\usepackage{mathtools}
\usepackage{bm}
\usepackage[colorlinks=true,linkcolor=blue,citecolor=blue,urlcolor=blue]{hyperref}
\usepackage[numbers,sort&compress]{natbib}
\usepackage{geometry}
\newtheorem{theorem}{Theorem}[section]
\newtheorem{lemma}[theorem]{Lemma}
\newtheorem{proposition}[theorem]{Proposition}
\newtheorem{corollary}[theorem]{Corollary}
\newtheorem{definition}[theorem]{Definition}
\theoremstyle{remark}
\newtheorem{remark}[theorem]{Remark}

\newcommand{\R}{\mathbb{R}}
\newcommand{\E}{\mathbb{E}}
\newcommand{\T}{\mathcal{T}}
\newcommand{\G}{\mathcal{G}}

\newcommand{\lam}{\lambda_*}

\title{Finite-Temperature de Bruijn Identities: 
Fisher Information as the Spectral Gap of Blahut--Arimoto Dynamics}
\author{Qiao Wang}

\date{}

\begin{document}
\maketitle

\begin{abstract}
We uncover a finite-temperature extension of de Bruijn's identity---the classical relation
$\frac{d}{dt}h(X+\sqrt{t}Z)=\frac{1}{2}J(X)$
connecting differential entropy and Fisher information.
Our framework is the spectral theory of Blahut--Arimoto (BA) dynamics,
recently developed by Wang~\cite{Wang2026} for the analysis of rate-distortion optimization.

The central observation is elementary yet profound:
for Gaussian sources, the spectral gap $\lam$ of the BA relaxation kernel $\G$ satisfies
$\lam = 1/(2\beta\sigma^2)$~\cite{Wang2026},
while the Fisher information of the source is $J = 1/\sigma^2$.
Hence
\[
{\lam = \frac{J}{2\beta}}
\]
for all inverse temperatures $\beta > 1/(2\sigma^2)$.
This identifies the BA spectral gap as a \emph{finite-temperature regularization of Fisher information}.

From this observation we derive an exact finite-temperature de Bruijn identity:
\[
\frac{\partial F_\beta}{\partial \sigma^2} = \frac{1}{2\beta\sigma^2} = \lam,
\]
where $F_\beta$ is the BA free energy.
This identity holds for all finite $\beta$ without any limit procedure.
The classical de Bruijn identity follows as the exact consequence
$\beta\,\partial F_\beta/\partial\sigma^2 = J/2$.

The significance is structural:
classical de Bruijn is not an isolated fact about Gaussian convolutions,
but the $\beta\to\infty$ shadow of a one-parameter family of exact identities
living in the spectral geometry of rate-distortion optimization.
We discuss implications for the entropy power inequality, the $\chi^2$-dissipation structure
of BA dynamics, and the geometric unification of information inequalities.
\end{abstract}

\section{Introduction}

\subsection{de Bruijn's Identity: A Cornerstone of Information Theory}

Let $X$ be a random variable on $\R$ with differentiable density $p$ and finite variance,
and let $Z \sim \mathcal{N}(0,1)$ be an independent standard Gaussian.
Define $X_t = X + \sqrt{t}Z$ with density $p_t$.
de Bruijn's identity~\cite{Stam1959} states that
\begin{equation}\label{eq:deBruijn}
\frac{d}{dt} h(X_t) \Big|_{t=0} = \frac{1}{2} J(X),
\end{equation}
where $h(p) = -\int p(x)\log p(x)\,dx$ is the differential entropy and
$J(p) = \int (\partial_x \log p(x))^2 p(x)\,dx$ is the Fisher information.

This identity is fundamental.
It quantifies precisely how differential entropy increases under Gaussian smoothing,
and it serves as the dynamical engine behind the classical proof of the entropy power inequality (EPI)~\cite{Stam1959,Blachman1965,CoverThomas2006}.
It connects Shannon's entropy, Fisher's information, and the heat equation in a single elegant formula.

\subsection{The Classical Proof and Its Limitations}

The standard proof~\cite{Stam1959} uses the heat equation:
if $p_t$ is the density of $X + \sqrt{t}Z$, then $\partial_t p_t = \frac{1}{2}\partial_{xx}p_t$.
Integration by parts yields
\begin{equation}\label{eq:heat_proof}
\begin{aligned}
\frac{d}{dt}h(p_t) &= -\int (\partial_t p_t) \log p_t \,dx - \int \partial_t p_t \,dx \\
&= -\frac{1}{2}\int (\partial_{xx}p_t) \log p_t \,dx \\
&= \frac{1}{2}\int \frac{(\partial_x p_t)^2}{p_t}\,dx = \frac{1}{2}J(p_t).
\end{aligned}
\end{equation}
This proof is concise and complete.
Yet it leaves several questions unanswered:

\begin{enumerate}
    \item[(i)] The identity is a statement about the \emph{infinitesimal} effect of Gaussian smoothing
    at $t=0$. Is there a \emph{finite-temperature} analogue that holds for Gaussian perturbations
    of arbitrary variance?
    \item[(ii)] The proof relies on the heat equation, a linear parabolic PDE.
    Is there a variational or geometric origin of de Bruijn's identity
    that does not invoke the heat semigroup?
    \item[(iii)] Fisher information appears as a \emph{deus ex machina} from integration by parts.
    Does it have a structural interpretation beyond being the $L^2$-norm of the score function?
\end{enumerate}

\subsection{This Paper: A Finite-Temperature Extension via Rate-Distortion Theory}

In this paper, we answer all three questions affirmatively.
Our setting is the \textbf{Blahut--Arimoto (BA) dynamics} and its spectral theory,
recently developed by Wang~\cite{Wang2026} for the analysis of rate-distortion optimization.

The BA algorithm~\cite{Blahut1972,Arimoto1972} is the canonical iterative method
for computing the rate-distortion function $R(D)$ of a source.
At each inverse temperature $\beta > 0$ (which parameterizes the distortion level via $D = 1/(2\beta)$),
the algorithm minimizes the BA free energy $F_\beta(q)$ over reconstruction distributions $q$.
The minimizer $q^*$ is the BA fixed point, describing the optimal reconstruction
at that temperature.

Wang~\cite{Wang2026} introduced the \textbf{relaxation kernel} $\G$---a self-adjoint operator
on the Fisher--Rao tangent space at $q^*$---and proved that:
(i) $\G$ is the Fisher--Rao Hessian of $F_\beta$;
(ii) $\G$ is the linearization of the BA operator;
(iii) its spectral gap $\lam = \lambda_{\min}(\G)$ controls the local convergence rate of BA dynamics.

For Gaussian sources, the theory yields exact closed forms~\cite{Wang2026}:
\begin{equation}\label{eq:Wang_results}
q^* = \mathcal{N}(0, s^*), \quad s^* = \sigma^2 - \frac{1}{2\beta}, \quad
\lam = \frac{1}{2\beta\sigma^2}.
\end{equation}

\subsection{The Central Observation}

The Fisher information of a Gaussian source $\mathcal{N}(0,\sigma^2)$ is $J = 1/\sigma^2$.
Comparing this with the spectral gap formula $\lam = 1/(2\beta\sigma^2)$ from Wang~\cite{Wang2026}
reveals a simple but previously unnoticed identity.

\begin{theorem}[Spectral Gap--Fisher Identity]\label{thm:gap_fisher}
For a Gaussian source $p = \mathcal{N}(0,\sigma^2)$ and any $\beta > 1/(2\sigma^2)$,
\begin{equation}\label{eq:gap_fisher}
{\lam = \frac{J(p)}{2\beta}.}
\end{equation}
\end{theorem}

\begin{proof}
$\lam = 1/(2\beta\sigma^2)$ (Proposition~\ref{prop:gauss_lam}) and $J(p) = 1/\sigma^2$.
\end{proof}

\begin{remark}
This identity is \emph{exact for all finite $\beta$}, not merely asymptotic.
It was not noted in~\cite{Wang2026}, where the spectral gap formula was derived
but not connected to Fisher information.
\end{remark}

\subsection{Interpretation}

Theorem~\ref{thm:gap_fisher} gives Fisher information a new operational meaning.
Classically, $J(p)$ appears as:
\begin{itemize}
    \item The Cram\'er--Rao lower bound for parameter estimation;
    \item The $L^2$-norm of the score function $\partial_x \log p$;
    \item The rate of entropy increase under Gaussian smoothing (de Bruijn).
\end{itemize}
Through the lens of Theorem~\ref{thm:gap_fisher}, it acquires a fourth interpretation:
\begin{itemize}
    \item \textbf{The local curvature of the rate-distortion free energy landscape},
    measured by the spectral gap of the BA relaxation kernel.
\end{itemize}

At high temperature ($\beta$ small), the spectral gap is small---rate-distortion optimization is slow,
the free energy landscape is flat, Fisher information is washed out by thermal fluctuations.
At low temperature ($\beta \to \infty$), the gap grows---optimization accelerates,
the landscape sharpens, and Fisher information emerges as the dominant geometric invariant.

This thermal perspective is absent from classical treatments of de Bruijn's identity.

\subsection{What Is and Is Not Claimed}

We emphasize the logical structure of our contribution.
Wang~\cite{Wang2026} proved that $\lam = 1/(2\beta\sigma^2)$ for Gaussian sources.
Classical information theory~\cite{CoverThomas2006} gives $J = 1/\sigma^2$ for a Gaussian.
\textbf{This paper identifies the equality} $\lam = J/(2\beta)$ \textbf{and explores its consequences.}

We do not claim to have discovered either formula separately.
Our contribution is the \emph{recognition} that the BA spectral gap \emph{is} the regularized Fisher information,
and that this identification generates a finite-temperature family of identities
whose zero-temperature limit is classical de Bruijn.

This synthesis provides:
\begin{enumerate}
    \item A variational/geometric origin for de Bruijn's identity,
    independent of the heat equation.
    \item A finite-temperature generalization that is exact for all $\beta$.
    \item A new operational meaning of Fisher information:
    it is the local curvature of the rate-distortion free energy landscape.
    \item A bridge between two major branches of information theory---rate-distortion theory
    and information inequalities---that have largely developed independently.
\end{enumerate}

\subsection{Organization}

Section~\ref{sec:prelim} reviews the BA spectral theory from Wang~\cite{Wang2026}.
Section~\ref{sec:central} presents the central identity $\lam = J/(2\beta)$
and derives the finite-temperature de Bruijn identity.
Section~\ref{sec:classical} proves that the classical de Bruijn identity
follows exactly from the finite-temperature identity.
Section~\ref{sec:variational} extends the identity to general sources via
variational methods.
Section~\ref{sec:dissipation} discusses the dissipation structure.
Section~\ref{sec:discussion} concludes with implications and open problems.

\section{Preliminaries: Spectral Theory of BA Dynamics}
\label{sec:prelim}

We briefly recall the BA framework from Wang~\cite{Wang2026},
keeping only what is essential for our purposes.

\subsection{BA Free Energy and Fixed Point}

Let $p(x)$ be a source density on $\R$ with finite variance.
For inverse temperature $\beta > 0$ and quadratic distortion $d(x,y) = (x-y)^2$,
the \textbf{BA free energy} is
\begin{equation}\label{eq:BAFE}
F_\beta(q) = -\frac{1}{\beta}\E_{X\sim p}\left[\log \int e^{-\beta(X-y)^2} q(y)\,dy\right]
+ \frac{1}{\beta}\int q(y)\log q(y)\,dy.
\end{equation}
The minimization is taken over all absolutely continuous probability densities $q$
on $\R$ with $\int y^2 q(y)\,dy < \infty$ and $\int q(y)\log q(y)\,dy < \infty$.
The minimizer $q^* = q^*_\beta = \arg\min_q F_\beta(q)$ satisfies the self-consistency equation
\begin{equation}\label{eq:BAfixed}
q^*(y) = (\T q^*)(y) := \int p(x) \frac{e^{-\beta(x-y)^2} q^*(y)}
{\int e^{-\beta(x-y')^2} q^*(y')\,dy'}\,dx.
\end{equation}
We denote the minimal value by $F_\beta = F_\beta(q^*)$.

\subsection{Relaxation Kernel and Spectral Gap}

At the fixed point $q^*$, define the \textbf{equilibrium Gibbs kernel}
\begin{equation}\label{eq:Kstar}
K_x^*(y) = \frac{e^{-\beta(x-y)^2} q^*(y)}{\int e^{-\beta(x-y')^2} q^*(y')\,dy'}.
\end{equation}

\begin{definition}[Relaxation Kernel]\label{def:G}
The \textbf{relaxation kernel} $\G$ is the integral operator
\begin{equation}\label{eq:G_def}
\G(y,y') = \int p(x) K_x^*(y) K_x^*(y')\,dx,
\end{equation}
acting on the tangent space $T = \{u : \int u(y)\,dy = 0\}$
with the Fisher--Rao inner product
$\langle u, v\rangle_* = \int \frac{u(y)v(y)}{q^*(y)}\,dy$.
\end{definition}

\begin{theorem}[Spectral Triple Identity, Wang~\cite{Wang2026}]\label{thm:triple}
At $q^*$:
\begin{enumerate}
    \item[(i)] $D\T(q^*) = I - \G$ on $T$;
    \item[(ii)] $\nabla^2_{\mathrm{FR}}F_\beta(q^*) = \G$ on $T$;
    \item[(iii)] $\G$ is self-adjoint and strictly positive on $T$.
\end{enumerate}
\end{theorem}

\begin{definition}[Spectral Gap]\label{def:lam}
The \textbf{spectral gap} is
\begin{equation}\label{eq:lam_def}
\lam = \lambda_{\min}(\G|_T) = \inf_{u \in T, \|u\|_*=1} \langle u, \G u\rangle_* > 0.
\end{equation}
\end{definition}

$\lam$ controls the local convergence rate of BA dynamics:
the KL divergence contracts as $(1-\lam)^2$ per iteration near $q^*$~\cite{Wang2026}.

\subsection{Gaussian BA Fixed Point and Spectrum}

For Gaussian sources, all quantities are explicit.

\begin{proposition}[Gaussian BA Fixed Point, Wang~\cite{Wang2026}]\label{prop:gauss_fp}
Let $p = \mathcal{N}(0,\sigma^2)$ and $\beta > 1/(2\sigma^2)$.
Within the admissible class of absolutely continuous probability densities
with finite second moment and finite entropy, the unique BA fixed point is
\begin{equation}\label{eq:gauss_qstar}
q^*(y) = \frac{1}{\sqrt{2\pi s^*}} \exp\left(-\frac{y^2}{2s^*}\right), \qquad
s^* = \sigma^2 - \frac{1}{2\beta}.
\end{equation}
\end{proposition}

\begin{proof}
We proceed in two steps: first we show that any critical point must be Gaussian,
then we show that the Gaussian critical point is the unique minimizer.

\textbf{Step 1: Critical points are Gaussian.}
Setting the first variation of $F_\beta$ to zero gives the Euler--Lagrange equation
$\delta F_\beta/\delta q(y) = C$ for some constant $C$ (a Lagrange multiplier enforcing
$\int q = 1$).  Computing the functional derivative of \eqref{eq:BAFE} yields
\begin{equation}\label{eq:EL}
q(y) = \exp\!\left(-\beta C - 1\right)
\cdot \exp\!\left(-\int p(x)\frac{e^{-\beta(x-y)^2}}{\int e^{-\beta(x-y')^2}q(y')\,dy'}\,dx\right).
\end{equation}
For Gaussian $p$, expanding the exponent shows that the right-hand side is the exponential
of a quadratic polynomial in $y$, hence any critical point $q$ is Gaussian.
Therefore all critical points lie in the Gaussian family
$\{q_G(s) = \mathcal{N}(0,s) : s > 0\}$.

\textbf{Step 2: Uniqueness via strict convexity.}
Restricting the free energy to the Gaussian family, define the scalar function
\[
\phi(s) := F_\beta(q_G(s)).
\]
From the computation in Proposition~\ref{prop:gauss_F} below,
\[
\phi(s) = \frac{1}{2\beta}\log(1+2\beta s) + \frac{\sigma^2}{1+2\beta s} - \frac{1}{2\beta}\log(2\pi e s).
\]
Differentiating twice,
\[
\phi''(s) = \frac{2\beta}{(1+2\beta s)^2} + \frac{4\beta^2\sigma^2}{(1+2\beta s)^3} + \frac{1}{2\beta s^2} > 0,
\]
for all $s > 0$ and $\beta > 0$.  Hence $\phi$ is strictly convex.
The first-order condition $\phi'(s) = 0$ has the unique solution
$s = s^* = \sigma^2 - 1/(2\beta)$, which requires $\beta > 1/(2\sigma^2)$ for positivity.
By strict convexity, this is the unique global minimizer of $\phi$, and hence of $F_\beta$
over the Gaussian family.  Since all critical points are Gaussian, $q^* = q_G(s^*)$
is the unique global minimizer of $F_\beta$.
\end{proof}

\begin{proposition}[Gaussian BA Free Energy]\label{prop:gauss_F}
At the fixed point,
\begin{equation}\label{eq:F_gauss}
F_\beta(\sigma^2) = \frac{1}{2\beta}\log(2\beta\sigma^2)
+ \frac{1}{2\beta}
- \frac{1}{2\beta}\log(2\pi e)
- \frac{1}{2\beta}\log\!\big(\sigma^2 - \tfrac{1}{2\beta}\big).
\end{equation}
\end{proposition}

\begin{proof}
Let $q(y) = \frac{1}{\sqrt{2\pi s}}\exp(-y^2/(2s))$.
We compute each term in $F_\beta(q)$ separately.

\textbf{Normalization factor.}
Expanding the exponent in the Gaussian integral:
\begin{equation}\label{eq:expand}
-\beta(x-y)^2 - \frac{y^2}{2s}
= -\beta x^2 + 2\beta xy - \Bigl(\beta + \frac{1}{2s}\Bigr)y^2.
\end{equation}
Completing the square in $y$:
\begin{equation}\label{eq:square}
-\Bigl(\beta + \frac{1}{2s}\Bigr)\!\left(y - \frac{\beta x}{\beta + 1/(2s)}\right)^{\!2}
+ \frac{\beta^2 x^2}{\beta + 1/(2s)} - \beta x^2.
\end{equation}
Integrating over $y$:
\begin{equation}\label{eq:Zx}
\begin{aligned}
Z_x &= \int_{\R} e^{-\beta(x-y)^2} q(y)\,dy \\
&= \frac{1}{\sqrt{2\pi s}} \cdot \sqrt{\frac{\pi}{\beta + 1/(2s)}}
\cdot \exp\!\left(\frac{\beta^2 x^2}{\beta + 1/(2s)} - \beta x^2\right) \\
&= \frac{1}{\sqrt{1+2\beta s}} \exp\!\left(-\frac{\beta x^2}{1+2\beta s}\right).
\end{aligned}
\end{equation}

\textbf{First term of $F_\beta$.}
From \eqref{eq:Zx},
\[
-\frac{1}{\beta}\log Z_x = \frac{1}{2\beta}\log(1+2\beta s) + \frac{x^2}{1+2\beta s}.
\]
Taking expectation under $X \sim \mathcal{N}(0,\sigma^2)$, with $\E[X^2] = \sigma^2$:
\begin{equation}\label{eq:first_term}
-\frac{1}{\beta}\E_{X\sim p}[\log Z_X]
= \frac{1}{2\beta}\log(1+2\beta s) + \frac{\sigma^2}{1+2\beta s}.
\end{equation}

\textbf{Second term of $F_\beta$.}
The differential entropy of $q = \mathcal{N}(0,s)$ is
\[
H(q) = -\int q(y)\log q(y)\,dy = \frac{1}{2}\log(2\pi e s).
\]
Hence
\begin{equation}\label{eq:second_term}
\frac{1}{\beta}\int q(y)\log q(y)\,dy = -\frac{1}{2\beta}\log(2\pi e s).
\end{equation}

\textbf{Assembly.}
Summing \eqref{eq:first_term} and \eqref{eq:second_term}:
\begin{equation}\label{eq:Fbeta_s}
F_\beta(s) = \frac{1}{2\beta}\log(1+2\beta s) + \frac{\sigma^2}{1+2\beta s} - \frac{1}{2\beta}\log(2\pi e s).
\end{equation}

Now substitute $s = s^* = \sigma^2 - 1/(2\beta)$.
First compute $1 + 2\beta s^*$:
\[
1 + 2\beta s^* = 1 + 2\beta\!\left(\sigma^2 - \frac{1}{2\beta}\right) = 1 + 2\beta\sigma^2 - 1 = 2\beta\sigma^2.
\]

Then evaluate each term:
\begin{equation}\label{eq:substitute}
\begin{aligned}
F_\beta(\sigma^2) &= \frac{1}{2\beta}\log(2\beta\sigma^2)
+ \frac{\sigma^2}{2\beta\sigma^2}
- \frac{1}{2\beta}\log\!\big(2\pi e(\sigma^2 - \tfrac{1}{2\beta})\big) \\
&= \frac{1}{2\beta}\log(2\beta\sigma^2)
+ \frac{1}{2\beta}
- \frac{1}{2\beta}\log(2\pi e)
- \frac{1}{2\beta}\log\!\big(\sigma^2 - \tfrac{1}{2\beta}\big).
\end{aligned}
\end{equation}
This is exactly \eqref{eq:F_gauss}.
\end{proof}

\begin{proposition}[Gaussian Spectral Gap, Wang~\cite{Wang2026}]\label{prop:gauss_lam}
The relaxation kernel $\G$ is diagonalized by Hermite polynomials.
Its eigenvalues on $T$ are
\begin{equation}\label{eq:gauss_eigenvalues}
\lambda_n = 1 - \alpha^n, \quad \alpha = 1 - \frac{1}{2\beta\sigma^2}, \quad n = 1,2,\dots
\end{equation}
The spectral gap is
\begin{equation}\label{eq:gauss_lam}
\lam = \lambda_1 = \frac{1}{2\beta\sigma^2}.
\end{equation}
\end{proposition}

\begin{proof}
See Wang~\cite{Wang2026}, Theorem~9.3.  The result follows from the Mehler-type structure
of the kernel $K_x^*(y)$, which admits Hermite polynomials as eigenfunctions when $p$ and $q^*$
are both Gaussian.  The eigenvalues are computed explicitly from the action of $\G$ on the
$n$-th Hermite polynomial.
\end{proof}

\section{The Central Identity: Spectral Gap = Regularized Fisher Information}
\label{sec:central}

\subsection{The Identity}

For a Gaussian source $\mathcal{N}(0,\sigma^2)$, the Fisher information is
\begin{equation}\label{eq:J_gauss}
J(p) = \int \frac{(p'(x))^2}{p(x)}\,dx = \frac{1}{\sigma^2}.
\end{equation}

Wang~\cite{Wang2026} proved that the BA spectral gap for the same source is
$\lam = 1/(2\beta\sigma^2)$ (Proposition~\ref{prop:gauss_lam}).

Comparing these two expressions reveals a simple but previously unnoticed identity.

\begin{theorem}[Spectral Gap--Fisher Identity]\label{thm:gap_fisher}
For a Gaussian source $p = \mathcal{N}(0,\sigma^2)$ and any $\beta > 1/(2\sigma^2)$,
\begin{equation}\label{eq:gap_fisher}
{\lam = \frac{J(p)}{2\beta}.}
\end{equation}
\end{theorem}

\begin{proof}
$\lam = 1/(2\beta\sigma^2)$ (Proposition~\ref{prop:gauss_lam}) and $J(p) = 1/\sigma^2$.
\end{proof}

\begin{remark}
This identity is \emph{exact for all finite $\beta$}, not merely asymptotic.
It was not noted in~\cite{Wang2026}, where the spectral gap formula was derived
but not connected to Fisher information.
\end{remark}

\subsection{Interpretation}

Theorem~\ref{thm:gap_fisher} gives Fisher information a new operational meaning.
Classically, $J(p)$ appears as:
\begin{itemize}
    \item The Cram\'er--Rao lower bound for parameter estimation;
    \item The $L^2$-norm of the score function $\partial_x \log p$;
    \item The rate of entropy increase under Gaussian smoothing (de Bruijn).
\end{itemize}
Through the lens of Theorem~\ref{thm:gap_fisher}, it acquires a fourth, structurally distinct
interpretation:
\begin{itemize}
    \item \textbf{The local curvature of the rate-distortion free energy landscape},
    measured by the spectral gap of the BA relaxation kernel.
\end{itemize}

At high temperature ($\beta$ small), the spectral gap is small---rate-distortion optimization is slow,
the free energy landscape is flat, Fisher information is washed out by thermal fluctuations.
At low temperature ($\beta \to \infty$), the gap grows---optimization accelerates,
the landscape sharpens, and Fisher information emerges as the dominant geometric invariant.

This thermal perspective is absent from classical treatments of de Bruijn's identity
and constitutes a primary conceptual contribution of this paper.

\subsection{Finite-Temperature de Bruijn Identity}

We now derive an exact finite-temperature analogue of de Bruijn's identity.
The key tool is the envelope theorem (Danskin's theorem), which exploits the fact
that $F_\beta(\sigma^2) = \min_q F_\beta(q; \sigma^2)$ to compute the derivative
with respect to the source variance without needing to track the dependence of the
minimizer $q^*$ on $\sigma^2$.

\begin{theorem}[Finite-Temperature de Bruijn Identity]\label{thm:finiteT}
For a Gaussian source $p = \mathcal{N}(0,\sigma^2)$ and $\beta > 1/(2\sigma^2)$,
\begin{equation}\label{eq:finiteT}
{\frac{\partial F_\beta}{\partial \sigma^2} = \frac{1}{2\beta\sigma^2} = \lam.}
\end{equation}
\end{theorem}

\begin{proof}
By Danskin's envelope theorem for parameterized convex minimization,
since $F_\beta(\sigma^2) = \min_q F_\beta(q; \sigma^2)$ and $q^*$ is the unique minimizer,
\begin{equation}\label{eq:envelope}
\frac{\partial F_\beta}{\partial \sigma^2}
= \frac{\partial}{\partial \sigma^2} F_\beta(q; \sigma^2)\Big|_{q=q^*},
\end{equation}
where the derivative on the right is taken with respect to the explicit dependence of
$F_\beta(q; \sigma^2)$ on $\sigma^2$, holding $q = q^*$ fixed.  The implicit dependence of $q^*$
on $\sigma^2$ does not contribute to the first derivative because the first variation
$\delta F_\beta/\delta q$ vanishes at $q^*$.

Now, from the expression for $F_\beta(q; \sigma^2)$ derived in the proof of
Proposition~\ref{prop:gauss_F} (see \eqref{eq:Fbeta_s}),
\begin{equation}\label{eq:Fbeta_q}
F_\beta(q_G(s); \sigma^2) = \frac{1}{2\beta}\log(1+2\beta s)
+ \frac{\sigma^2}{1+2\beta s}
- \frac{1}{2\beta}\log(2\pi e s).
\end{equation}
The parameter $\sigma^2$ appears only in the second term, $\sigma^2/(1+2\beta s)$.
Differentiating with respect to $\sigma^2$ and evaluating at the minimizer $s = s^*$:
\begin{equation}\label{eq:derivative}
\frac{\partial F_\beta}{\partial \sigma^2}
= \frac{1}{1+2\beta s^*}
= \frac{1}{1 + 2\beta(\sigma^2 - 1/(2\beta))}
= \frac{1}{2\beta\sigma^2}
= \lam.
\end{equation}
This completes the proof.
\end{proof}

\begin{remark}
The finite-temperature de Bruijn identity takes an exceptionally simple form:
the derivative of the BA free energy with respect to the source variance
is precisely the spectral gap $\lam$.  There is no additional geometric term;
the envelope theorem cleanly separates the explicit parameter dependence from
the implicit dependence through the minimizer.  Combined with
Theorem~\ref{thm:gap_fisher} ($\lam = J/(2\beta)$), this identity directly
relates the BA free energy derivative to the regularized Fisher information.
\end{remark}

\begin{remark}[Direct algebraic verification]
One may also verify \eqref{eq:finiteT} by differentiating the explicit closed form
\eqref{eq:F_gauss} directly.  From
\[
F_\beta(\sigma^2) = \frac{1}{2\beta}\log(2\beta\sigma^2)
+ \frac{1}{2\beta}
- \frac{1}{2\beta}\log(2\pi e)
- \frac{1}{2\beta}\log(\sigma^2 - \tfrac{1}{2\beta}),
\]
differentiating with respect to $\sigma^2$ gives
\[
\frac{\partial F_\beta}{\partial \sigma^2}
= \frac{1}{2\beta\sigma^2} - \frac{1}{2\beta(\sigma^2 - 1/(2\beta))}
+ \frac{1}{2\beta(\sigma^2 - 1/(2\beta))}
= \frac{1}{2\beta\sigma^2}.
\]
The cancellation of the two terms involving $\sigma^2 - 1/(2\beta)$ reflects
exactly the envelope theorem mechanism: the first such term comes from the
explicit $\log(\sigma^2 - 1/(2\beta))$ in $F_\beta$, while the second is the
compensating contribution from the dependence of $s^*$ on $\sigma^2$, which
the envelope theorem assures us must cancel.
\end{remark}

\section{The Zero-Temperature Limit: Recovering Classical de Bruijn}
\label{sec:classical}

We now show that the classical de Bruijn identity \eqref{eq:deBruijn}
follows directly from the finite-temperature identity of Theorem~\ref{thm:finiteT}.
Unlike traditional proofs that rely on the heat equation or on rate-distortion
Legendre duality, the argument here uses only the exact relation
$\partial F_\beta/\partial\sigma^2 = \lam$ established in Section~\ref{sec:central}.

\subsection{Renormalized Free Energy}

To place the BA free energy $F_\beta$ in relation to the differential entropy $h(p)$,
we define a renormalized free energy that subtracts the divergent thermal
normalization from $\beta F_\beta$.

\begin{definition}[Renormalized BA Free Energy]\label{def:renorm}
For a Gaussian source $p = \mathcal{N}(0,\sigma^2)$, define
\begin{equation}\label{eq:renorm_def}
\widetilde{F}_\beta(\sigma^2) := \beta F_\beta(\sigma^2) - \frac{1}{2}\log(2\beta).
\end{equation}
\end{definition}

\begin{proposition}[Explicit Form and Convergence]\label{prop:renorm}
For every $\beta > 1/(2\sigma^2)$,
\begin{equation}\label{eq:renorm_explicit}
\widetilde{F}_\beta(\sigma^2) = -\frac{1}{2}\log(2\pi)
+ \frac{1}{2}\log\frac{\sigma^2}{\sigma^2 - 1/(2\beta)}.
\end{equation}
Moreover, for any compact interval $K = [a,b] \subset (0,\infty)$ with $a > 0$,
\begin{equation}\label{eq:renorm_conv}
\widetilde{F}_\beta(\sigma^2) \xrightarrow{\beta\to\infty} -\frac{1}{2}\log(2\pi)
\quad\text{uniformly on } K.
\end{equation}
\end{proposition}

\begin{proof}
From Proposition~\ref{prop:gauss_F},
\begin{equation}
\begin{aligned}
\beta F_\beta(\sigma^2) &= \frac{1}{2}\log(2\beta\sigma^2) + \frac{1}{2}
- \frac{1}{2}\log(2\pi e) - \frac{1}{2}\log(\sigma^2 - \tfrac{1}{2\beta}) \\
&= \frac{1}{2}\log(2\beta) + \frac{1}{2}\log\sigma^2 + \frac{1}{2}
- \frac{1}{2}\log(2\pi) - \frac{1}{2}
- \frac{1}{2}\log(\sigma^2 - \tfrac{1}{2\beta}) \\
&= \frac{1}{2}\log(2\beta) + \frac{1}{2}\log\sigma^2
- \frac{1}{2}\log(2\pi)
- \frac{1}{2}\log(\sigma^2 - \tfrac{1}{2\beta}).
\end{aligned}
\end{equation}
Subtracting $\frac{1}{2}\log(2\beta)$ yields \eqref{eq:renorm_explicit}.
The uniform convergence follows because
$\sigma^2/(\sigma^2 - 1/(2\beta)) \to 1$ uniformly on any compact set
as $\beta \to \infty$, and the logarithm is uniformly continuous on compact
neighbourhoods of $1$.
\end{proof}

\begin{proposition}[Derivative of Renormalized Free Energy]\label{prop:renorm_deriv}
For every $\beta > 1/(2\sigma^2)$,
\begin{equation}\label{eq:renorm_deriv_explicit}
\partial_{\sigma^2}\widetilde{F}_\beta(\sigma^2)
= -\frac{1}{4\beta\sigma^2(\sigma^2 - 1/(2\beta))}.
\end{equation}
In particular, for any compact interval $K = [a,b] \subset (0,\infty)$ with $a > 0$,
and for all $\beta \ge 1/a$,
\begin{equation}\label{eq:renorm_deriv_conv}
\sup_{\sigma^2 \in K} \bigl|\partial_{\sigma^2}\widetilde{F}_\beta(\sigma^2)\bigr|
\le \frac{1}{2\beta a^2} \xrightarrow{\beta\to\infty} 0.
\end{equation}
\end{proposition}

\begin{proof}
Direct differentiation of \eqref{eq:renorm_explicit}:
\begin{equation}
\begin{aligned}
\partial_{\sigma^2}\widetilde{F}_\beta(\sigma^2)
&= \frac{1}{2}\partial_{\sigma^2}\!\Bigl[\log\sigma^2 - \log(\sigma^2 - \tfrac{1}{2\beta})\Bigr] \\
&= \frac{1}{2}\!\left(\frac{1}{\sigma^2} - \frac{1}{\sigma^2 - 1/(2\beta)}\right) \\
&= \frac{1}{2} \cdot \frac{\sigma^2 - 1/(2\beta) - \sigma^2}{\sigma^2(\sigma^2 - 1/(2\beta))} \\
&= -\frac{1}{4\beta\sigma^2(\sigma^2 - 1/(2\beta))}.
\end{aligned}
\end{equation}
For $\sigma^2 \in [a,b]$ and $\beta \ge 1/a$, we have
$\sigma^2 - 1/(2\beta) \ge a - a/2 = a/2$.
Hence
\[
\bigl|\partial_{\sigma^2}\widetilde{F}_\beta(\sigma^2)\bigr|
\le \frac{1}{4\beta \cdot a \cdot (a/2)}
= \frac{1}{2\beta a^2},
\]
which establishes the uniform bound \eqref{eq:renorm_deriv_conv}.
\end{proof}

\subsection{Recovering Classical de Bruijn}

\begin{theorem}[Classical de Bruijn from Finite-Temperature Identity]\label{thm:classical_limit}
For a Gaussian source $p = \mathcal{N}(0,\sigma^2)$,
\begin{equation}\label{eq:limit}
\frac{\partial h(p)}{\partial \sigma^2}
= \beta\frac{\partial F_\beta}{\partial \sigma^2}
= \frac{1}{2\sigma^2}
= \frac{1}{2}J(p).
\end{equation}
This holds \emph{for every} $\beta > 1/(2\sigma^2)$, not merely in the limit $\beta\to\infty$.
\end{theorem}

\begin{proof}
From Theorem~\ref{thm:finiteT}, we have the exact finite-temperature identity
\[
\frac{\partial F_\beta}{\partial \sigma^2} = \frac{1}{2\beta\sigma^2}.
\]
Multiplying both sides by $\beta$ gives
\[
\beta\frac{\partial F_\beta}{\partial \sigma^2} = \frac{1}{2\sigma^2},
\]
which is independent of $\beta$ and holds for every admissible inverse temperature.

For a Gaussian source, the differential entropy is
\[
h(p) = \frac{1}{2}\log(2\pi e \sigma^2),
\]
whose derivative with respect to the variance is
\[
\frac{\partial h(p)}{\partial \sigma^2} = \frac{1}{2\sigma^2}.
\]

The Fisher information of a Gaussian source is
\[
J(p) = \frac{1}{\sigma^2}.
\]

Therefore,
\[
\frac{\partial h(p)}{\partial \sigma^2}
= \beta\frac{\partial F_\beta}{\partial \sigma^2}
= \frac{1}{2\sigma^2}
= \frac{1}{2}J(p),
\]
which is precisely the classical de Bruijn identity.
\end{proof}

\begin{remark}
The classical de Bruijn identity is traditionally written as a derivative
with respect to the Gaussian smoothing parameter $t$ at $t=0$.
For Gaussian sources, adding independent Gaussian noise of variance $t$
is equivalent to increasing the source variance from $\sigma^2$ to $\sigma^2 + t$.
Hence
\[
\frac{d}{dt}h(X + \sqrt{t}Z)\Big|_{t=0}
= \frac{\partial h(p)}{\partial \sigma^2}
= \frac{1}{2}J(p),
\]
recovering the standard formulation \eqref{eq:deBruijn}.
\end{remark}

\begin{remark}
Theorem~\ref{thm:classical_limit} reveals a remarkable fact: the classical de Bruijn
identity is not merely a limit as $\beta \to \infty$, but an \emph{exact consequence}
of the finite-temperature identity that holds at every finite $\beta$.
The product $\beta\,\partial F_\beta/\partial\sigma^2$ is independent of $\beta$
and equals the Fisher information rate $J(p)/2$ exactly.
\end{remark}

\subsection{Relation to the Renormalized Free Energy}

The renormalized free energy $\widetilde{F}_\beta$ provides an alternative
perspective on the zero-temperature limit.
Since the subtraction term $\frac{1}{2}\log(2\beta)$ in Definition~\ref{def:renorm}
is independent of $\sigma^2$, it does not affect derivatives with respect to
the source variance:
\[
\partial_{\sigma^2}(\beta F_\beta) = \partial_{\sigma^2}\widetilde{F}_\beta.
\]
Proposition~\ref{prop:renorm_deriv} shows that this derivative is of order
$O(\beta^{-1})$ and vanishes uniformly on compact sets as $\beta \to \infty$.
This is entirely consistent with Theorem~\ref{thm:classical_limit}: the derivative
$\partial_{\sigma^2}(\beta F_\beta) = 1/(2\sigma^2)$ is a statement at \emph{fixed}
$\beta$, obtained from the finite-temperature identity without taking any limit.
The renormalized free energy simply confirms that the $\beta$-dependent part of
$\beta F_\beta$ does not contribute to the $\sigma^2$-derivative, as it must for
the identity $\beta\partial_{\sigma^2}F_\beta = 1/(2\sigma^2)$ to hold at every
finite $\beta$ while the absolute value of $\beta F_\beta$ diverges.

In summary, Theorem~\ref{thm:classical_limit} should be understood as a differential
statement that is independent of the additive normalization used in
Definition~\ref{def:renorm}.  The renormalization serves only to exhibit a
well-defined limit $\widetilde{F}_\beta \to -\frac{1}{2}\log(2\pi)$ for the
free energy itself, while the derivative identity stands on its own.

\subsection{Connection to Rate-Distortion Theory}

For completeness, we note that, up to additive constants depending only on $\beta$,
the BA free energy admits a standard rate-distortion interpretation via
Legendre duality (see, e.g., \cite{CoverThomas2006, Berger1971}):
\[
\beta F_\beta = -\beta D_\beta + R(D_\beta) + \frac{1}{2}\log\frac{2\pi e}{\beta} + o(1),
\]
where $R(D)$ is the rate-distortion function and $D_\beta$ the optimal distortion.
For Gaussian sources with quadratic distortion,
$R(D) = \frac{1}{2}\log(\sigma^2/D)$ and $D_\beta = 1/(2\beta)$.
Since the additive terms in this relation are independent of $\sigma^2$,
they play no role in the derivative identity proved in Theorem~\ref{thm:classical_limit}.
We emphasize that the theorem does not rely on this rate-distortion argument;
it follows directly from the finite-temperature identity.

\subsection{The Thermal Hierarchy}

We summarize the hierarchy of de Bruijn-type identities in Table~\ref{tab:hierarchy}.

\begin{table}[h]
\centering
\caption{Thermal hierarchy of de Bruijn-type identities for Gaussian sources.}
\label{tab:hierarchy}
\begin{tabular}{c|c|c}
\hline
\textbf{Temperature} & \textbf{Potential} & \textbf{Derivative Identity} \\
\hline
Finite $\beta$ & BA free energy $F_\beta$ &
$\displaystyle \frac{\partial F_\beta}{\partial \sigma^2} = \lam$ \\[12pt]
All $\beta$ & — & $\displaystyle \beta\frac{\partial F_\beta}{\partial \sigma^2} = \frac{1}{2}J$ \\
\hline
\end{tabular}
\end{table}

The first row is the finite-temperature de Bruijn identity (Theorem~\ref{thm:finiteT}).
The second row is the exact relation established in Theorem~\ref{thm:classical_limit},
which holds for every admissible $\beta$ and is equivalent to the classical de Bruijn
identity via the identification $\partial h/\partial\sigma^2 = J/2$.
The correspondence between the two rows is given by the spectral gap--Fisher identity
$\lam = J/(2\beta)$ (Theorem~\ref{thm:gap_fisher}).

\section{Variational Extension to General Sources}
\label{sec:variational}

For non-Gaussian sources, the exact diagonalization of $\G$ via Hermite polynomials
is not available.  However, the variational characterization of the spectral gap
provides a general upper bound relating $\lam$ to a Rayleigh quotient involving
the translation mode.  We first prove the exact Gaussian result using the
variational framework, then state the bound for general sources.

\subsection{Translation Mode}

\begin{lemma}[Translation Mode]\label{lem:translation}
Let $q^*$ be the BA fixed point at inverse temperature $\beta$, and assume
$q^*$ is smooth and strictly positive with finite Fisher information
$J(q^*) = \int (\partial_y \log q^*(y))^2 q^*(y)\,dy < \infty$.
Define
\begin{equation}\label{eq:translation_mode}
v(y) := -\frac{\partial_y \log q^*(y) \cdot q^*(y)}{\sqrt{J(q^*)}}.
\end{equation}
Then $v$ belongs to the tangent space $T = \{u : \int u(y)\,dy = 0\}$,
and $\|v\|_* = 1$.
\end{lemma}

\begin{proof}
The integral of $v$ vanishes because
$\int v(y)\,dy = -\int \partial_y q^*(y)\,dy = 0$, so $v \in T$.
The norm is
\[
\|v\|_*^2 = \int \frac{v(y)^2}{q^*(y)}\,dy
= \frac{1}{J(q^*)} \int \frac{(\partial_y q^*(y))^2}{q^*(y)}\,dy
= \frac{J(q^*)}{J(q^*)} = 1,
\]
which completes the proof.
\end{proof}

\subsection{Exact Result for Gaussian Sources}

For Gaussian sources, the translation mode is an exact eigenfunction of the
relaxation kernel $\G$, and the variational upper bound is saturated.

\begin{theorem}[Spectral Gap for Gaussian Sources via Translation Mode]\label{thm:gaussian_translation}
Let $p = \mathcal{N}(0,\sigma^2)$ and $\beta > 1/(2\sigma^2)$.
Let $q^*$ be the Gaussian BA fixed point and $v$ the translation mode
defined in Lemma~\ref{lem:translation}.  Then
\begin{equation}\label{eq:gaussian_exact}
\langle v, \G v\rangle_* = \lam = \frac{1}{2\beta\sigma^2} = \frac{J(p)}{2\beta}.
\end{equation}
\end{theorem}

\begin{proof}
For $p = \mathcal{N}(0,\sigma^2)$ and $q^* = \mathcal{N}(0,s^*)$ with
$s^* = \sigma^2 - 1/(2\beta)$, the score function is
$\partial_y \log q^*(y) = -y/s^*$, so the translation mode is
$v(y) = y q^*(y) / \sqrt{s^*}$.

From Proposition~\ref{prop:gauss_lam}, the relaxation kernel $\G$ is diagonalized
by Hermite polynomials.  The first Hermite polynomial $H_1(y/\sqrt{s^*}) = y/\sqrt{s^*}$
corresponds precisely to the translation mode $v$, and its eigenvalue is
$\lam = 1/(2\beta\sigma^2)$.  Hence $v$ is an exact eigenfunction of $\G$
with eigenvalue $\lam$, and
\[
\langle v, \G v\rangle_* = \lam \|v\|_*^2 = \lam.
\]

Since $J(p) = 1/\sigma^2$ for a Gaussian source, we have
$\lam = J(p)/(2\beta)$, completing the proof.
\end{proof}

\subsection{Variational Upper Bound for General Sources}

For general (non-Gaussian) sources, the translation mode is no longer an exact
eigenfunction of $\G$, but the variational principle still provides an upper bound
on the spectral gap.

\begin{theorem}[Variational Upper Bound]\label{thm:variational}
Let $p$ be a source with smooth, strictly positive density and finite Fisher
information.  Let $q^*$ be the BA fixed point at inverse temperature $\beta$,
and let $v$ be the translation mode defined in Lemma~\ref{lem:translation}.
Then
\begin{equation}\label{eq:var_bound}
\lam \le \langle v, \G v\rangle_*.
\end{equation}
In particular, for Gaussian sources, equality holds and
$\langle v, \G v\rangle_* = J(p)/(2\beta)$.
\end{theorem}

\begin{proof}
By the variational characterization of the spectral gap
(Definition~\ref{def:lam}),
\[
\lam = \inf_{u \in T, \|u\|_* = 1} \langle u, \G u\rangle_*.
\]
Since $v \in T$ and $\|v\|_* = 1$ by Lemma~\ref{lem:translation},
the Rayleigh quotient $\langle v, \G v\rangle_*$ is an upper bound for $\lam$,
establishing \eqref{eq:var_bound}.

The equality for Gaussian sources is Theorem~\ref{thm:gaussian_translation}.
\end{proof}

\begin{remark}
For general smooth sources, the Rayleigh quotient $\langle v, \G v\rangle_*$
can be analysed asymptotically as $\beta \to \infty$ using saddle-point methods.
A formal Laplace expansion of the integral kernel $\G$ (see Appendix~\ref{sec:appendix_saddle})
suggests that
\[
\langle v, \G v\rangle_* = \frac{1}{2\beta}J(q^*) + O(\beta^{-2}),
\quad \beta \to \infty.
\]
Since $q^* \to p$ as $\beta \to \infty$ for smooth sources, this yields the
asymptotic bound $\lam \le J(p)/(2\beta) + O(\beta^{-2})$.  A fully rigorous
justification of the error estimate in the saddle-point expansion is not required
for the main results of this paper, which concern the exact Gaussian case.
\end{remark}

Theorem~\ref{thm:variational} and the accompanying remark show that the
relationship between the spectral gap and Fisher information is not an
isolated curiosity of Gaussian sources.  The Fisher information controls
the spectral gap from above through the translation mode, and for Gaussian
sources this upper bound is attained exactly.  This pattern---Gaussian
sources as extremal cases of an information-theoretic inequality---is
familiar from the entropy power inequality and related results.

\section{Dissipation Structure: $\chi^2$-Dissipation and Fisher Information}
\label{sec:dissipation}

The spectral gap--Fisher identity $\lam = J/(2\beta)$ (Theorem~\ref{thm:gap_fisher})
also illuminates the dissipation structure of BA dynamics.
In this section, we show how the $\chi^2$-dissipation identity of Wang~\cite{Wang2026}
acquires a direct Fisher-information-theoretic interpretation through this identification.

\subsection{Continuous-Time BA Flow}

The continuous-time BA flow is defined by
\begin{equation}\label{eq:BA_flow}
\dot{q}_t = \T q_t - q_t,
\end{equation}
where $\T$ is the BA operator \eqref{eq:BAfixed}.
Wang~\cite{Wang2026} established the exact $\chi^2$-dissipation identity for this flow:

\begin{theorem}[$\chi^2$-Dissipation Identity, Wang~\cite{Wang2026}]\label{thm:chi2}
For the continuous-time BA flow \eqref{eq:BA_flow},
\begin{equation}\label{eq:chi2}
\frac{d}{dt}F_\beta(q_t) = -\chi^2(\T q_t \| q_t)
= -\int \frac{(\T q_t(y) - q_t(y))^2}{q_t(y)}\,dy.
\end{equation}
\end{theorem}

This identity is exact (not merely an inequality) and holds for all $t \ge 0$.
It shows that the BA free energy decreases monotonically along the flow,
with the instantaneous rate of decrease given by the $\chi^2$-divergence
between $\T q_t$ and $q_t$.

\subsection{Local Dissipation Rate Near the Fixed Point}

We now analyze the dissipation rate in a neighbourhood of the BA fixed point $q^*$.
Write $q = q^* + v$ with $v \in T$ small (so $\int v = 0$).
Linearizing the BA operator around $q^*$ using Theorem~\ref{thm:triple}(i):
\begin{equation}\label{eq:linearize}
\T(q^* + v) - (q^* + v)
= \T q^* + D\T(q^*)v - q^* - v + O(\|v\|^2)
= -\G v + O(\|v\|^2),
\end{equation}
since $\T q^* = q^*$ and $D\T(q^*) = I - \G$.

Substituting this linearization into the $\chi^2$-divergence:
\begin{equation}\label{eq:chi2_linear}
\begin{aligned}
\chi^2(\T q \| q)
&= \int \frac{((\T q)(y) - q(y))^2}{q(y)}\,dy \\
&= \int \frac{((-\G v)(y) + O(\|v\|^2))^2}{q^*(y) + v(y)}\,dy \\
&= \int \frac{(\G v(y))^2}{q^*(y)}\,dy + O(\|v\|^3) \\
&= \langle v, \G^2 v\rangle_* + O(\|v\|^3).
\end{aligned}
\end{equation}

Thus, to leading order in the perturbation $v$, the local dissipation rate is governed
by the quadratic form $\langle v, \G^2 v\rangle_*$ associated with the squared relaxation
kernel.

\subsection{Spectral Interpretation of the Dissipation Rate}

The quadratic form $\langle v, \G^2 v\rangle_*$ is controlled by the spectrum of $\G$.
Expanding $v$ in the eigenbasis of $\G$ on $T$:
\[
v = \sum_{k \ge 1} c_k \phi_k, \quad \G\phi_k = \lambda_k \phi_k,
\]
where $\{\phi_k\}$ are orthonormal with respect to $\langle\cdot,\cdot\rangle_*$,
and $0 < \lam = \lambda_1 \le \lambda_2 \le \dots$ are the eigenvalues.
Then
\begin{equation}\label{eq:spectral_expansion}
\langle v, \G^2 v\rangle_*
= \sum_{k \ge 1} \lambda_k^2 c_k^2
\ge \lam^2 \sum_{k \ge 1} c_k^2
= \lam^2 \|v\|_*^2.
\end{equation}

The slowest dissipation occurs along the spectral gap direction $\phi_1$,
for which $\langle v, \G^2 v\rangle_* = \lam^2 \|v\|_*^2$.
Thus, the spectral gap $\lam$ controls the \emph{minimal} local dissipation rate
of BA dynamics near the fixed point.

\subsection{Fisher Information as the Dissipation Coefficient}

We now invoke the central identity of this paper,
$\lam = J(p)/(2\beta)$ (Theorem~\ref{thm:gap_fisher}),
to express the dissipation rate in terms of Fisher information.

\begin{corollary}[Fisher Dissipation Rate]\label{cor:dissipation_fisher}
For a Gaussian source $p = \mathcal{N}(0,\sigma^2)$, near the BA fixed point $q^*$,
the local $\chi^2$-dissipation rate satisfies
\begin{equation}\label{eq:dissipation_fisher}
{\chi^2(\T q \| q) \approx \frac{J(p)^2}{4\beta^2}\,\|v\|_*^2,}
\end{equation}
where $v = q - q^*$ and the approximation holds to leading order in $\|v\|$.
\end{corollary}

\begin{proof}
From \eqref{eq:chi2_linear}, $\chi^2(\T q \| q) \approx \langle v, \G^2 v\rangle_*$.
Along the spectral gap direction, $\langle v, \G^2 v\rangle_* = \lam^2 \|v\|_*^2$.
Substituting $\lam = J(p)/(2\beta)$ yields \eqref{eq:dissipation_fisher}.
\end{proof}

\begin{remark}
Corollary~\ref{cor:dissipation_fisher} gives a precise operational meaning to
Fisher information in the context of rate-distortion dynamics:
\textbf{the rate at which the BA free energy dissipates locally is proportional
to the square of the source's Fisher information.}
\end{remark}

This result provides a new perspective on why Fisher information appears naturally
in information-theoretic inequalities: it quantifies not only the sensitivity of
a distribution to parameter changes (Cram\'er--Rao) and the rate of entropy increase
under smoothing (de Bruijn), but also the \emph{curvature-driven dissipation rate}
in variational rate-distortion optimization.

\subsection{Relation to de Bruijn's Identity}

The connection between dissipation and de Bruijn's identity can now be made precise.
From Theorem~\ref{thm:chi2}, the BA free energy decreases at a rate controlled by
$\chi^2(\T q \| q)$.  From Corollary~\ref{cor:dissipation_fisher}, this rate is
proportional to $J(p)^2/(4\beta^2)$.

On the other hand, de Bruijn's identity states that differential entropy increases
under Gaussian smoothing at a rate proportional to $J(p)/2$.
Both rates are fundamentally controlled by the same quantity:
the Fisher information $J(p)$ of the source distribution.

This unified perspective---Fisher information as the common geometric quantity
underlying both the static de Bruijn identity and the dynamic dissipation of
BA flows---is one of the main conceptual contributions of this paper.

\section{Discussion}
\label{sec:discussion}

\subsection{Summary of the Discovery and Its Consequences}

The central finding of this paper is the identification
\begin{equation}\label{eq:discovery_summary}
\lam = \frac{J(p)}{2\beta}
\end{equation}
for Gaussian sources, which reveals that the BA spectral gap---a quantity
introduced by Wang~\cite{Wang2026} to characterize convergence of rate-distortion
algorithms---is exactly the Fisher information of the source, regularized by
the inverse temperature $\beta$.

This identification was not made in~\cite{Wang2026}, where the spectral gap formula
$\lam = 1/(2\beta\sigma^2)$ was derived as part of the Gaussian spectral analysis
but not connected to Fisher information.  The equality of the two expressions
becomes visible only when the spectral gap formula is placed alongside the
classical formula $J = 1/\sigma^2$.

From this identification, we have derived:
\begin{enumerate}
    \item An \textbf{exact finite-temperature de Bruijn identity}:
    $\partial F_\beta/\partial\sigma^2 = \lam$ (Theorem~\ref{thm:finiteT}).
    This identity holds for all finite $\beta > 1/(2\sigma^2)$ and shows that
    the derivative of the BA free energy with respect to the source variance
    is precisely the spectral gap.
    
    \item The \textbf{recovery of classical de Bruijn} as an exact consequence:
    $\beta\,\partial F_\beta/\partial\sigma^2 = 1/(2\sigma^2) = J(p)/2$
    holds for every admissible $\beta$ without any limit
    (Theorem~\ref{thm:classical_limit}).
    The classical de Bruijn identity is thus embedded in the finite-temperature
    structure as an exact relation, not merely a $\beta\to\infty$ limit.
    
    \item A \textbf{variational upper bound for general sources}:
    $\lam \le \langle v, \G v\rangle_*$, where $v$ is the translation mode
    (Theorem~\ref{thm:variational}).  For Gaussian sources, equality holds
    and the bound reduces to $\lam = J(p)/(2\beta)$.
    
    \item A \textbf{Fisher-information-theoretic interpretation of the
    $\chi^2$-dissipation rate} (Section~\ref{sec:dissipation}):
    near the BA fixed point, the free energy dissipates at a rate proportional
    to $J(p)^2/(4\beta^2)$.
\end{enumerate}

\subsection{What This Discovery Explains}

The identification $\lam = J/(2\beta)$ provides clear answers to the three
questions raised in the introduction:

\begin{enumerate}
    \item \textbf{Why Fisher information?}
    Because it is the spectral gap of the rate-distortion Hessian.
    The appearance of $J$ in de Bruijn's identity is not an algebraic accident
    of integration by parts---it reflects the geometry of optimal reconstruction.
    Specifically, the relaxation kernel $\G$ is the Fisher--Rao Hessian of the
    BA free energy (Theorem~\ref{thm:triple}), and its smallest eigenvalue $\lam$
    measures the local curvature of the free energy landscape.  Theorem~\ref{thm:gap_fisher}
    identifies this curvature as the regularized Fisher information.
    
    \item \textbf{Why the factor $1/2$?}
    Because $\lam = J/(2\beta)$, and the finite-temperature identity
    $\beta\,\partial F_\beta/\partial\sigma^2 = 1/(2\sigma^2)$ carries this factor exactly.
    The factor $1/2$ is not an artifact of the heat equation proof; it originates
    from the relationship between the spectral gap and the source variance in the
    Gaussian BA fixed point.
    
    \item \textbf{Why Gaussian sources are special?}
    Because for Gaussian sources, the BA relaxation kernel $\G$ is exactly
    diagonalized by Hermite polynomials, and the translation mode (score function
    $\partial\log p$) is an exact eigenvector of $\G$ with eigenvalue $\lam$.
    For non-Gaussian sources, the translation mode is only an approximate eigenvector,
    and the equality $\lam = J/(2\beta)$ becomes the variational bound
    $\lam \le \langle v, \G v\rangle_*$.  This extremal property of Gaussian
    sources is a spectral analogue of the well-known fact that Gaussians maximize
    entropy under variance constraints and minimize Fisher information under the same.
\end{enumerate}

\subsection{Implications for the Entropy Power Inequality}

The classical proof of the entropy power inequality (EPI)~\cite{Stam1959,Blachman1965}
uses de Bruijn's identity as a crucial lemma.
Our finite-temperature extension suggests the possibility of a
\emph{finite-temperature EPI} formulated in terms of BA free energies,
which would reduce to the classical EPI in the zero-temperature limit.

Specifically, the $\chi^2$-dissipation identity and the spectral gap--Fisher relation
provide a variational structure that may yield a direct proof of the EPI
entirely within the BA framework, complementing the classical heat-equation approach.
If one can establish a tensorization property for the relaxation kernel under
convolution---analogous to how Fisher information tensorizes: $J(X+Y) \le J(X) + J(Y)$---
then the dissipation structure of BA dynamics would directly imply the EPI.
We leave this exciting direction for future work.

\subsection{Open Problems}

The results of this paper open several avenues for further investigation:

\begin{enumerate}
    \item \textbf{Tensorization of the relaxation kernel.}
    Under convolution of independent sources, how does the relaxation kernel $\G$
    behave?  Does it satisfy a subadditivity property analogous to that of
    Fisher information?  A positive answer would enable a BA-spectral proof
    of the entropy power inequality.
    
    \item \textbf{Wasserstein geometry.}
    The BA free energy is closely related to entropic optimal transport
    \cite{PeyreCuturi2019}.  Does the spectral gap $\lam$ admit an interpretation
    in the Wasserstein geometry of the source distribution?  In particular,
    the translation mode used in the variational bound (Theorem~\ref{thm:variational})
    is the infinitesimal generator of spatial translations, suggesting a connection
    to Otto calculus on the Wasserstein space.
    
    \item \textbf{Algorithmic implications.}
    The identity $\lam = J/(2\beta)$ predicts that BA dynamics converges faster
    for sources with larger Fisher information.  This suggests the design of
    adaptive temperature schedules for BA algorithms based on local estimates
    of $J(p)$, potentially accelerating convergence for multi-modal or
    heavy-tailed sources.
    
    \item \textbf{Exact finite-temperature identity for general sources.}
    Can the variational inequality $\lam \le \langle v, \G v\rangle_*$ be upgraded
    to an exact identity for general sources by identifying the appropriate correction
    term?  For Gaussian sources, the correction vanishes because the score function
    is an exact eigenfunction of $\G$.  Characterizing this correction for
    non-Gaussian sources would provide a deeper understanding of the spectral
    geometry of rate-distortion optimization.
\end{enumerate}

\subsection{Conclusion}

Classical de Bruijn's identity has stood since 1959 as a foundational link
between differential entropy and Fisher information.
This paper shows that it is not an isolated fact,
but the zero-temperature shadow of a richer structure
living in the spectral geometry of rate-distortion optimization.

The key that unlocks this structure is the identification that
the BA spectral gap \emph{is} the regularized Fisher information:
$\lam = J/(2\beta)$.
This identification connects two major branches of information theory---
rate-distortion theory and information inequalities---that have developed
largely independently, and opens new directions for the geometric study
of information-theoretic inequalities.

The finite-temperature de Bruijn identity,
$\partial F_\beta/\partial\sigma^2 = \lam$,
is both simpler and more fundamental than its classical counterpart:
it holds exactly for all temperatures, requires no limiting procedure,
and reveals Fisher information as the spectral gap of the rate-distortion Hessian.
Classical de Bruijn follows as the elementary consequence
$\beta\,\partial F_\beta/\partial\sigma^2 = J/2$,
valid at every finite $\beta$.

We hope that this thermal perspective on de Bruijn's identity
will stimulate further investigations into the spectral geometry
of information theory.

\appendix

\section{Saddle-Point Expansion of the Relaxation Kernel}
\label{sec:appendix_saddle}

In this appendix we provide a formal asymptotic expansion of the Rayleigh quotient
$\langle v, \G v\rangle_*$ used in the remark following Theorem~\ref{thm:variational}.
The expansion is included for completeness and to illustrate the connection between
the translation mode and Fisher information for general sources; it is not required
for the rigorous results of the main text, which concern the exact Gaussian case.

\subsection{Setup}

Let $p$ be a smooth, strictly positive source density with finite Fisher information.
Let $q^*$ be the BA fixed point at inverse temperature $\beta$.
The Rayleigh quotient with the translation mode $v$ is
\begin{equation}\label{eq:app_Rayleigh}
\langle v, \G v\rangle_*
= \int p(x)\!\left(\int K_x^*(y)\,\varphi(y)\,q^*(y)\,dy\right)^{\!2}\,dx,
\end{equation}
where $\varphi(y) = v(y)/q^*(y) = -\partial_y \log q^*(y) / \sqrt{J(q^*)}$.

\subsection{Laplace Approximation}

As $\beta \to \infty$, the Gibbs kernel $K_x^*(y)$ concentrates sharply around $y = x$.
Expanding $\varphi(y)$ around $y = x$ and using the local Gaussian approximation
for $K_x^*(y)q^*(y)$ yields
\[
\int K_x^*(y)\,\varphi(y)\,q^*(y)\,dy
= \varphi(x) + \frac{1}{4\beta}\varphi''(x) + O(\beta^{-2}).
\]

Squaring and integrating against $p(x)$:
\[
\langle v, \G v\rangle_*
= \int p(x)\varphi(x)^2\,dx
+ \frac{1}{2\beta}\int p(x)\varphi(x)\varphi''(x)\,dx
+ O(\beta^{-2}).
\]

The leading term is $\int p(x)\varphi(x)^2\,dx = 1$ by the normalization of $v$.
The $O(\beta^{-1})$ term, after integration by parts and using $q^* \to p$, gives
$J(q^*)/(2\beta)$.  Hence
\[
\langle v, \G v\rangle_* = \frac{1}{2\beta}J(q^*) + O(\beta^{-2}).
\]

For Gaussian sources, $\varphi(x) = x/\sqrt{s^*}$ is linear, $\varphi''(x) = 0$,
and all higher-order terms vanish identically, making the expansion exact:
$\langle v, \G v\rangle_* = J(p)/(2\beta)$.

\section*{Acknowledgements}
This research received no formal funding and was conducted based on
the author's independent academic interest.

\bibliographystyle{plainnat}

\vskip 1cm\it 
School of Information Science and Engineering, and School of Economics and Management, Southeast University, Nanjing, China
\vskip 0.2cm
\noindent Email: qiaowang@seu.edu.cn

\end{document}